\documentclass[11pt]{article}
\usepackage{algorithm}
\usepackage{algorithmic}

\usepackage{comment}

\newtheorem{theorem}{Theorem}[section]
\newtheorem{proposition}[theorem]{Proposition}
\newtheorem{lemma}[theorem]{Lemma}
\newtheorem{corollary}[theorem]{Corollary}

\def\squarebox#1{\hbox to #1{\hfill\vbox to #1{\vfill}}}
\def\qed{\hspace*{\fill}%
        \vbox{\hrule\hbox{\vrule\squarebox{.667em}\vrule}\hrule}\smallskip}
\newenvironment{proof}{\begin{trivlist}
\item[\hspace{\labelsep}{\em\noindent Proof.~}]}{\qed\end{trivlist}}

\title{Strategy-Proof Approximation Algorithms for the Stable Marriage Problem with Ties and Incomplete Lists}
\author{Koki Hamada$^{1}$, Shuichi Miyazaki$^{2}$, Hiroki Yanagisawa$^{3}$
\\
\small $^{1}$NTT  Corporation.\\
\small \texttt{hamada.koki@lab.ntt.co.jp}\\
\small $^{2}$Kyoto University.\\
\small \texttt{shuichi@media.kyoto-u.ac.jp}\\
\small $^{3}$IBM Research - Tokyo.\\
\small \texttt{yanagis@jp.ibm.com}}

\date{}

\begin{document}

\sloppy

\maketitle

\begin{abstract}
In the stable marriage problem (SM), a mechanism that always outputs a stable matching is called a {\em stable mechanism}.  One of the well-known stable mechanisms is the man-oriented Gale-Shapley algorithm (MGS).  MGS has a good property that it is strategy-proof to the men's side, i.e., no man can obtain a better outcome by falsifying a preference list.  We call such a mechanism a {\em man-strategy-proof mechanism}.  Unfortunately, MGS is not a woman-strategy-proof mechanism.  (Of course, if we flip the roles of men and women, we can see that the woman-oriented Gale-Shapley algorithm (WGS) is a woman-strategy-proof but not a man-strategy-proof mechanism.)  Roth has shown that there is no stable mechanism that is simultaneously man-strategy-proof and woman-strategy-proof, which is known as Roth's impossibility theorem.

In this paper, we extend these results to the stable marriage problem with ties and incomplete lists (SMTI).  Since SMTI is an extension of SM, Roth's impossibility theorem takes over to SMTI.  Therefore, we focus on the one-sided-strategy-proofness.  In SMTI, one instance can have stable matchings of different sizes, and it is natural to consider the problem of finding a largest stable matching, known as MAX SMTI.  Thus we incorporate the notion of approximation ratio used in the theory of approximation algorithms.  We say that a stable-mechanism is {\em $c$-approximate-stable mechanism} if it always returns a stable matching of size at least $1/c$ of a largest one.  We also consider a  restricted variant of MAX SMTI, which we call MAX SMTI-1TM, where only men's lists can contain ties (and women's lists must be strictly ordered).

Our results are summarized as follows: (i) MAX SMTI admits both a man-strategy-proof $2$-approximate-stable mechanism and a woman-strategy-proof $2$-approximate-stable mechanism.  (ii) MAX SMTI-1TM admits a woman-strategy-proof 2-approximate-stable mechanism.  (iii) MAX SMTI-1TM admits a man-strategy-proof 1.5-approximate-stable mechanism.  All these results are tight in terms of approximation ratios.  Also, all these strategy-proofness results apply for strategy-proofness against coalitions.
\end{abstract}

\section{Introduction}\label{sec:intro}

An instance of the {\em stable marriage problem} ({\em SM}) \cite{gs62} consists of $n$ men $m_{1}, m_{2}, \ldots, m_{n}$, $n$ women $w_{1}, w_{2}, \ldots, w_{n}$, and each person's preference list, which is a total order of all the members of the opposite gender.  If a person $q_{i}$ precedes a person $q_{j}$ in a person $p$'s preference list, then we write $q_{i} \succ_{p} q_{j}$ and interpret it as ``$p$ prefers $q_{i}$ to $q_{j}$''.  In this paper, we denote a preference list in the following form:

\smallskip

$m_{2}$: $w_{3}$ \ $w_{1}$ \ $w_{4}$ \ $w_{2}$,

\smallskip

\noindent
which means that $m_{2}$ prefers $w_{3}$ best, $w_{1}$ 2nd, $w_{4}$ 3rd, and $w_{2}$ last (this example is for $n=4$).

A {\em matching} is a set of $n$ (man, woman)-pairs in which no person appears more than once.  For a matching $M$, $M(p)$ denotes the partner of a person $p$ in $M$.  If, for a man $m$ and a woman $w$, both $w \succ_{m} M(m)$ and $m \succ_{w} M(w)$ hold, then we say that $(m,w)$ is a {\em blocking pair} for $M$ or $(m,w)$ {\em blocks} $M$.  Note that both $m$ and $w$ have incentive to be matched with each other ignoring the given partner, so it can be thought of as a threat for the current matching $M$.  A matching with no blocking pair is a {\em stable matching}.  It is known that any instance admits at least one stable matching, and one can be found by the {\em Gale-Shapley algorithm} (or {\em GS algorithm} for short) in $O(n^{2})$ time \cite{gs62}.  There have been a plenty of research results on this problem from viewpoints of Economics, Computer Science, Mathematics, etc (see \cite{gi89, rs90, man13} e.g.).

\subsection{Strategy-proofness}

The stable marriage problem can be seen as a game among participants, who have true preferences in mind, but may submit a falsified preference list hoping to obtain a better partner than the one assigned when true preference lists are used.  Formally, let $S$ be a {\em mechanism}, that is, a mapping from instances to matchings, and we denote $S(I)$ the matching output by $S$ for an instance $I$.  We say that $S$ is a {\em stable mechanism} if, for any instance $I$, $S(I)$ is a stable matching for $I$.  For a mechanism $S$, let $I$ be an instance, $M$ be a matching such that $M=S(I)$, and $p$ be a person.  We say that $p$ {\em has a successful strategy in $I$} if there is an instance $I'$ in which people except for $p$ have the same preference lists in $I$ and $I'$, and $p$ prefers $M'$ to $M$ (i.e., $M'(p) \succ_{p} M(p)$ with respect to $p$'s preference list in $I$), where $M'$ is a matching such that $M'=S(I')$.  This situation is interpreted as follows: $I$ is the set of true preference lists, and by submitting a falsified preference list (which changes the set of lists to $I'$), $p$ can obtain a better partner $M'(p)$.  We say that $S$ is a {\em strategy-proof mechanism} if, when $S$ is used, no person has a successful strategy in any instance.  Also we say that $S$ is a {\em man-strategy-proof mechanism} if, when $S$ is used, no man has a successful strategy in any instance.  A {\em woman-strategy-proof mechanism} is defined analogously.  A mechanism is a {\em one-sided-strategy-proof mechanism} if it is either a man-strategy-proof mechanism or a woman-strategy-proof mechanism.

It is known that there is no strategy-proof stable mechanism for SM \cite{roth82}, which is known as {\em Roth's impossibility theorem}.  By contrast, the man-oriented GS algorithm, {\em MGS} for short, (in which men send and women receive proposals; see Appendix \ref{sec:MGS}) is a man-strategy-proof stable mechanism for SM \cite{roth82, df81}.  Of course, by the symmetry of men and women, the woman-oriented GS algorithm ({\em WGS}) is a woman-strategy-proof stable mechanism.

\subsection{Ties and Incomplete Lists}

One of the most natural extensions of SM is the {\em Stable Marriage with Ties and Incomplete lists}, denoted {\em SMTI}.  An instance of SMTI consists of  $n$ men, $n$ women, and each person's preference list.  A preference list may include ties, which represent indifference between two or more persons, and may be incomplete, meaning that a preference list may contain only a {\em subset} of people in the opposite gender.  Such a preference list may be of the following form:

\smallskip

$m_{2}$: $w_{3}$ \ ($w_{1}$ \ $w_{4}$),

\smallskip

\noindent
which represents that $m_{2}$ prefers $w_{3}$ best, $w_{1}$ and $w_{4}$ 2nd with equal preference, but does not want to be matched with $w_{2}$.  If a person $q$ is included in $p$'s preference list, we say that $q$ is {\em acceptable} to $p$.  A {\em matching} is a set of mutually acceptable (man, woman)-pairs in which no person appears more than once.  The {\em size} of a matching $M$, denoted $|M|$, is the number of pairs in $M$.  For a matching $M$, $(m,w)$ is a {\em blocking pair} if (i) $m$ and $w$ are acceptable to each other, (ii) $m$ is single in $M$ or $w \succ_{m} M(m)$, and (iii) $w$ is single in $M$ or $m \succ_{w} M(w)$.  A matching without blocking pairs is a {\em stable matching}. (When ties come into consideration, there are three definitions for stability, {\em super}, {\em strong}, and {\em weak} stabilities.  Here we are considering weak stability which is the most natural notion among the three.  In the case of super and strong stabilities, there exist instances that do not admit a stable matching.  See \cite{gi89, man13} for more details.)  

Note that in the case of SM, the size of a matching is always $n$ by definition, but it may be less than $n$ in the case of SMTI.  In fact, there is an SMTI-instance that admits stable matchings of different sizes, and the problem of finding a maximum size stable matching, called {\em MAX SMTI}, is NP-hard \cite{immm99, miimm02}.  There are a plenty of approximability and inapproximability results for MAX SMTI.  The current best upper bound on the approximation ratio is 1.5 \cite{m09,p11,k13} and lower bounds are $33/29 \simeq 1.1379$ assuming P$\neq$NP and $4/3 \simeq 1.3333$ assuming the Unique Games Conjecture (UGC) \cite{yanag07}.  There are several attempts to obtain better algorithms (e.g., polynomial-time exact algorithms or polynomial-time approximation algorithms with better approximation ratio) for restricted instances; one of the most natural restrictions is to admit ties in preference lists of only one gender, which we call {\em SMTI-1T}.  MAX SMTI-1T (i.e., the problem of finding a maximum cardinality stable matching in SMTI-1T) remains NP-hard, and as for the approximation ratio, the current best upper bound is $1+1/e \simeq 1.368$ \cite{lp19} and lower bounds are $21/19 \simeq 1.1052$ assuming P$\neq$NP and $5/4 = 1.25$ assuming UGC \cite{himy07, yanag07}.

\subsection{Our Contributions}\label{subsec:ourresults}

In this paper, we consider the strategy-proofness in MAX SMTI, and investigate the trade-off between strategy-proofness and approximability.  In the case of incomplete preference lists, there may be unmatched (i.e., single) persons.  Thus, we have to extend the definition of a person preferring one matching to another.  We say that a person $p$ prefers $M'$ to $M$ if either $M'(p) \succ_{p} M(p)$ holds or $p$ is single in $M$ but is matched in $M'$ with some acceptable woman.  Then the definition of strategy-proofness for SM naturally takes over to SMTI.

Let $I$ be a MAX SMTI instance and $M_{opt}$ be a maximum size stable matching for $I$.  A stable matching $M$ for $I$ is called an {\em $r$-approximate solution for $I$} if $\frac{|M_{opt}|}{|M|} \leq r$.  A stable mechanism $S$ is called an {\em $r$-approximate-stable mechanism} if $S(I)$ is an $r$-approximate solution for any MAX SMTI instance $I$.

Firstly, since SMTI is a generalization of SM, Roth's impossibility theorem for SM \cite{roth82} holds also for MAX SMTI (regardless of approximability):

\begin{proposition}
There is no strategy-proof stable mechanism for MAX SMTI.
\end{proposition}

Hence we focus on {\em one-sided}-strategy-proofness.  We show that there is a 2-approximate-stable mechanism, which is achieved by a simple extension of the GS algorithm.  We also show that this result is tight:

\begin{theorem}\label{thm:SMTI}
MAX SMTI admits both a man-strategy-proof $2$-approximate-stable mechanism and a woman-strategy-proof $2$-approximate-stable mechanism.  On the other hand, for any positive $\epsilon$, MAX SMTI admits neither a man-strategy-proof $(2-\epsilon)$-approximate-stable mechanism nor a woman-strategy-proof $(2-\epsilon)$-approximate-stable mechanism.
\end{theorem}

We next consider a restricted version, MAX SMTI-1T.  Throughout the paper, we assume that ties appear in men's lists only (and women's lists must be strict).  In the following, we write MAX SMTI-1T as {\em MAX SMTI-1TM} to stress that only men's preference lists may contain ties.  As for woman-strategy-proofness, we obtain the same result as for MAX SMTI, which is a direct consequence of Theorem \ref{thm:SMTI}:

\begin{corollary}\label{coro:SMTI-1TM}
MAX SMTI-1TM admits a woman-strategy-proof 2-approximate-stable mechanism, but no woman-strategy-proof $(2-\epsilon)$-approximate-stable mechanism for any positive $\epsilon$.
\end{corollary}

For man-strategy-proofness, we can reduce the approximation ratio to 1.5, which is the main result of this paper.

\begin{theorem}\label{thm:SMTI-1TM}
MAX SMTI-1TM admits a man-strategy-proof 1.5-approximate-stable mechanism, but no man-strategy-proof $(1.5-\epsilon)$-approximate-stable mechanism for any positive $\epsilon$.
\end{theorem}

We remark that no assumptions on running times are made for our negative results, while algorithms in our positive results run in linear time.  Note also that the current best polynomial-time approximation algorithms for MAX SMTI and MAX SMTI-1TM have the approximation ratios better than those in our negative results (Theorems \ref{thm:SMTI} and \ref{thm:SMTI-1TM}).  Hence our results provide gaps between polynomial-time computation and strategy-proof computation.

\bigskip

\noindent
{\bf Coalition.} 
In the above discussion, the man-strategy-proofness (woman-strategy-proofness) is defined in terms of a manipulation of a preference list by one man (woman).  We can extend this notion to the {\em coalition} of men (or women) as follows; a coalition $C$ of men has a successful strategy if there is a way of falsifying preference lists of members of $C$ which improves the outcome of {\em every} member of $C$.  It is known that MGS is strategy-proof against the coalition of men in this sense (Theorem 1.7.1 of \cite{gi89}), and this strategy-proofness holds also in the stable marriage with incomplete lists (SMI) (page 57 of \cite{gi89}).  Since all our strategy-proofness results (Lemmas \ref{lm:positiveSMTI} and \ref{lm:strategy-proof}) are attributed to the strategy-proofness of MGS in SMI, we can easily modify the proofs so that Theorem \ref{thm:SMTI}, Corollary \ref{coro:SMTI-1TM}, and Theorem \ref{thm:SMTI-1TM} hold for strategy-proofness against coalitions.

\bigskip

\noindent
{\bf Overview of Techniques.}
Since MGS is a man-strategy-proof stable mechanism for SM, such types of algorithms are good candidates for proving the positive part of Theorem~\ref{thm:SMTI-1TM}.  Existing 1.5-approximation algorithms for MAX SMTI for one-sided ties are of GS-type, but in these algorithms, proposals are made from the side with no ties (women, in our case), so we cannot use them for our purpose.  As mentioned above, there are 1.5-approximation algorithms for the general MAX SMTI \cite{m09,p11,k13}, which are fortunately of GS-type and can handle proposals from the side with ties (men, in our case).  Hence one may expect that these algorithms will work.  However, it is not the case.  The main reason is as follows: Suppose that some man $m$ is going to propose to a woman, and the head of $m$'s current list is a tie, which is a mixture of unmatched and matched women.  In this case, $m$'s proposal will be sent to an unmatched woman, say $w$.  Suppose that, just one step before, another man $m'$ has proposed to $w'$.  Then if $m'$ moves $w$ to the position just before $w'$, he can make $w$ already matched when $m$ is about to propose to her, and as a result of this, $m$ does not propose to $w$ but to another unmatched woman.  In this way, a man can change another man's proposal order, which destroys the strategy-proofness (see Appendix \ref{sec:counter-ex} for more details).  To overcome it, we modify an existing 1.5-approximation algorithm to be {\em robust} in the sense that a man's proposal order is not affected by other men's preference lists.

\bigskip

\noindent
{\bf Ties {\em or} Incomplete Lists.}
When only ties are present (SMT) or only incomplete lists are present (SMI), all the stable matchings of one instance have the same cardinality.  The former is due to the fact that any stable matching is a perfect matching, and the latter is due to the Rural Hospitals theorem \cite{gs85}.  Hence approximability is not an important issue in these cases.  As for strategy-proofness, since SMT and SMI are generalizations of SM, Roth's impossibility theorem holds and no strategy-proof stable mechanism exists.  Existence of one-sided strategy-proofness for SMI is already known as we have mentioned in ``Coalition'' part above, and that for SMT follows directly from Theorem \ref{thm:SMTI}.


\section{Results for MAX SMTI}\label{sec:SMTIand1T}

In this section, we give a proof of Theorem \ref{thm:SMTI}.  We start with the positive part:

\begin{lemma}\label{lm:positiveSMTI}
MAX SMTI admits both a man-strategy-proof $2$-approximate-stable mechanism and a woman-strategy-proof $2$-approximate-stable mechanism.
\end{lemma}

\begin{proof}
Consider a mechanism $S^{*}$ that is described by the following algorithm.  Given a MAX SMTI instance $I$, $S^{*}$ first breaks each tie so that persons in a tie are ordered increasingly in their indices, that is, if $q_{i}$ and $q_{j}$ are in the same tie of $p$'s list, then after the tie break $q_{i} \succ_{p} q_{j}$ holds if and only if $i < j$.  (This ordering is only for the purpose of making the algorithm deterministic, so any {\em fixed} tie-breaking rule is valid.)  Let $I'$ be the resulting instance.  Its preference lists are incomplete but do not include ties; such an instance is called an {\em SMI instance}.  It then applies MGS modified for SMI \cite{gi89} to $I'$ and obtains a stable matching $M$ for $I'$.  It is easy to see that $M$ is stable for $I$.  Also it is well-known that in MAX SMTI, any stable matching is a 2-approximate solution \cite{miimm02}.  Hence $S^{*}$ is a 2-approximate-stable mechanism.

We then show that $S^{*}$ is a man-strategy-proof mechanism for MAX SMTI.  Suppose not.  Then there is a MAX SMTI instance $I$ and a man $m$ who has a successful strategy in $I$.  Let $J$ be a MAX SMTI instance in which only $m$'s preference list differs from $I$, and by using it $m$ obtains a better outcome.  Let $M_{I}$ and $M_{J}$ be the outputs of $S^{*}$ on $I$ and $J$, respectively.  Then $m$ prefers $M_{J}$ to $M_{I}$, that is, either (i) $M_{J}(m) \succ_{m} M_{I}(m)$ with respect to $m$'s true preference list in $I$, or (ii) $m$ is single in $M_{I}$ and matched in $M_{J}$, and $M_{J}(m)$ is acceptable to $m$ in $I$.  Let $I'$ and $J'$, respectively, be the SMI-instances constructed from $I$ and $J$ by breaking ties in the above mentioned manner.  Then $M_{I}$ and $M_{J}$ are, respectively, the results of MGS applied to $I'$ and $J'$.  Since $I'$ is the result of tie-breaking of $I$ and $m$ prefers $M_{J}$ to $M_{I}$ in $I$, $m$ prefers $M_{J}$ to $M_{I}$ in $I'$.  Note that, due to the tie-breaking rule, the preference lists of people except for $m$ are same in $I'$ and $J'$.  This means that when MGS is used in SMI, $m$ can have a successful strategy in $I'$ (i.e., to change his list to that of $J'$),contradicting to the man-strategy-proofness of MGS for SMI (page 57 of \cite{gi89}).

If we exchange the roles of men and women in $S^{*}$, we obtain a woman-strategy-proof 2-approximate-stable mechanism.
\end{proof}

We then show the negative part.  We remark that $\epsilon$ is not necessarily a constant.

\begin{lemma}\label{lm:negativeSMTI}
For any positive $\epsilon$, MAX SMTI admits neither a man-strategy-proof $(2-\epsilon)$-approximate-stable mechanism nor a woman-strategy-proof $(2-\epsilon)$-approximate-stable mechanism.
\end{lemma}

\begin{proof}
Consider the instance $I_{1}$ given in Fig.~\ref{fig:input1}, where $m_{3}$'s preference list is empty.  It is straightforward to verify that $I_{1}$ has two stable matchings $M_{1} = \{ (m_{1}, w_{1}), (m_{2}, w_{2}) \}$ and $M_{2} = \{ (m_{1}, w_{2}), (m_{2}, w_{3}) \}$, both of which are of maximum size.  Hence any stable mechanism must output either $M_{1}$ or $M_{2}$.

\begin{figure}[htb]
\begin{center}
\begin{tabular}{ccccccccccccc}
$m_{1}$: & $w_{2}$ & $w_{1}$ & & \hspace{1mm} & $w_{1}$: & $m_{1}$ & & \\
$m_{2}$: & $w_{2}$ & $w_{3}$ & & \hspace{1mm} & $w_{2}$: & ($m_{1}$ & $m_{2}$) & \\
$m_{3}$: & & & & \hspace{1mm} & $w_{3}$: & $m_{2}$ & &
\end{tabular}
\caption{A MAX SMTI instance $I_{1}$}\label{fig:input1}
\end{center}
\end{figure}

First, consider an arbitrary $(2-\epsilon)$-approximate-stable mechanism $S$ for MAX SMTI that outputs $M_{1}$ on $I_{1}$.  Let $I'_{1}$ be the instance obtained from $I_{1}$ by deleting $w_{1}$ from $m_{1}$'s preference list.  Then since $M_{2}$ is still a stable matching for $I'_{1}$ and $S$ is a $(2-\epsilon)$-approximate-stable mechanism, $S$ must output a stable matching of size 2.  But since $M_{2}$ is now the only stable matching of size 2, $S$ outputs $M_{2}$ on $I'_{1}$.  Thus $m_{1}$ can obtain a better partner by manipulating his preference list.  Next, consider a $(2-\epsilon)$-approximate-stable mechanism $S$ that outputs $M_{2}$ on $I_{1}$.  Then let $I''_{1}$ be the instance obtained from $I_{1}$ by deleting $w_{3}$ from $m_{2}$'s preference list.  By a similar argument, $S$ must output $M_{1}$ on $I''_{1}$ and hence $m_{2}$ can obtain a better partner by manipulation.  We have shown that, for any $(2-\epsilon)$-approximate-stable mechanism $S$, some man has a successful strategy in $I_{1}$ and hence $S$ is not a man-strategy-proof mechanism.

Next we use the instance $I_{2}$ given in Fig.~\ref{fig:input2}, which is symmetric to $I_{1}$.  By the same argument as above, we can show that for any $(2-\epsilon)$-approximate-stable mechanism $S$, some woman has a successful strategy in $I_{2}$ and hence $S$ is not a woman-strategy-proof mechanism.

\begin{figure}[htb]
\begin{center}
\begin{tabular}{ccccccccccccc}
$m_{1}$: & $w_{1}$ & & & \hspace{1mm} & $w_{1}$: & $m_{2}$ & $m_{1}$ & \\
$m_{2}$: & ($w_{1}$ & $w_{2}$) & & \hspace{1mm} & $w_{2}$: & $m_{2}$ & $m_{3}$ & \\
$m_{3}$: & $w_{2}$ & & & \hspace{1mm} & $w_{3}$: & & &
\end{tabular}
\caption{A MAX SMTI instance $I_{2}$}\label{fig:input2}
\end{center}
\end{figure}

Therefore, we have shown that any $(2-\epsilon)$-approximate-stable mechanism is neither a man-strategy-proof mechanism nor a woman-strategy-proof mechanism.
\end{proof}

\section{Results for MAX SMTI-1TM}\label{sec:SMTI-1TM}

Recall that, in a MAX SMTI-1TM instance, ties can appear in men's lists only.  Since $I_{2}$ in the proof of Lemma \ref{lm:negativeSMTI} is a MAX SMTI-1TM instance, the second half of the proof applies also to MAX SMTI-1TM.  Consequently, together with Lemma \ref{lm:positiveSMTI}, we have Corollary \ref{coro:SMTI-1TM}.

We then move to man-strategy-proofness and give a proof for Theorem \ref{thm:SMTI-1TM}.  We start with the negative part:

\begin{lemma}\label{lm:negativeSMTI-1TM}
For any positive $\epsilon$, there is no man-strategy-proof $(1.5-\epsilon)$-approximate-stable mechanism for MAX SMTI-1TM.
\end{lemma}

\begin{proof}
The proof goes like that of Lemma \ref{lm:negativeSMTI}.  Consider the instance $I_{3}$ in Fig.~\ref{fig:input3}.  
$I_{3}$ has four matchings of size 3, namely, $M_{3}= \{ (m_{1}, w_{1}), (m_{2}, w_{2}), (m_{3}, w_{3}) \}$, $M_{4}= \{ (m_{1}, w_{1}), (m_{2}, w_{2}), (m_{3}, w_{4}) \}$, $M_{5}= \{ (m_{1}, w_{1}), (m_{2}, w_{3}), (m_{3}, w_{4}) \}$, and $M_{6}= \{ (m_{1}, w_{2}), (m_{2}, w_{3}), (m_{3}, w_{4}) \}$.
Among them, $M_{3}$ and $M_{6}$ are stable ($M_{4}$ is blocked by $(m_{3}, w_{3})$ and $M_{5}$ is blocked by $(m_{1}, w_{2})$).  Hence any $(1.5-\epsilon)$-approximate-stable mechanism outputs either $M_{3}$ or $M_{6}$, since a stable matching of size 2 is not a $(1.5-\epsilon)$-approximate solution.

\begin{figure}[htb]
\begin{center}
\begin{tabular}{ccccccccccccc}
$m_{1}$: & $w_{2}$ & $w_{1}$ & & \hspace{1mm} & $w_{1}$: & $m_{1}$ & & \\
$m_{2}$: & ($w_{2}$ & $w_{3}$) & & \hspace{1mm} & $w_{2}$: & $m_{2}$ & $m_{1}$ & \\
$m_{3}$: & $w_{3}$ & $w_{4}$ & & \hspace{1mm} & $w_{3}$: & $m_{2}$ & $m_{3}$ & \\
$m_{4}$: & & & & \hspace{1mm} & $w_{4}$: & $m_{3}$ & &
\end{tabular}
\caption{A MAX SMTI-1TM instance $I_{3}$}\label{fig:input3}
\end{center}
\end{figure}

Consider an arbitrary $(1.5-\epsilon)$-approximate-stable mechanism $S$ for MAX SMTI-1TM, and suppose that  $S$ outputs $M_{3}$ on $I_{3}$.  Then if $m_{1}$ deletes $w_{1}$ from the list, $M_{6}$ is the unique maximum stable matching (of size 3); hence $S$ must output $M_{6}$ and so $m_{1}$ can obtain a better partner $w_{2}$.  Similarly, if $S$ outputs $M_{6}$ on $I_{3}$, $m_{3}$ can force $S$ to output $M_{3}$ by deleting $w_{4}$ from the list.  In either case, some man has a successful strategy in $I_{3}$ and hence $S$ is not a man-strategy-proof mechanism.
\end{proof}

Finally, we give a proof for the positive part, which is the main result of this paper.

\begin{lemma}\label{lm:positiveSMTI-1TM}
There exists a man-strategy-proof $1.5$-approximate-stable mechanism for MAX SMTI-1TM.
\end{lemma}

\begin{proof}
We give Algorithm~\ref{alg:pk6} and show that it is a man-strategy-proof $1.5$-approximate-stable mechanism by three subsequent lemmas (Lemmas \ref{lm:stability}--\ref{lm:strategy-proof}).
Algorithm~\ref{alg:pk6} first translates an SMTI-1TM instance $I$ to an SMI instance $I'$ using Algorithm~\ref{alg:convi}, then applies MGS to $I'$ and obtains a matching $M'$, and finally constructs a matching $M$ of $I$ from $M'$.  It is worth noting that a man $b_{j}$ created at Step \ref{step:convi6} of Algorithm~\ref{alg:convi} corresponds to a woman (not a man) of $I$.  As will be seen later, $b_{j}$ is definitely matched with $s_{j}$ or $t_{j}$ in $M'$, and the other woman (i.e., either $s_{j}$ or $t_{j}$ who is not matched with $b_{j}$) plays a role of woman $w_{j}$ of $I$: If she is single in $M'$, then $w_{j}$ is single in $M$.  If she is matched with $a_{i}$ in $M'$, then $w_{j}$ is matched with $m_{i}$ in $M$.


 \begin{algorithm}[htb]
  \caption{An algorithm for MAX SMTI-1TM}\label{alg:pk6}
  \begin{algorithmic}[1]
   \REQUIRE An instance $I$ for MAX SMTI-1TM.
   \ENSURE A matching $M$ for $I$.
   \STATE Construct an SMI instance $I'$ from $I$ using Algorithm~\ref{alg:convi}.
   \STATE Apply MGS to $I'$ and obtain a matching $M'$.\label{step:pk6gs}
   \STATE Let $M := \{(m_{i},w_{j}) \mid (a_{i}, s_{j}) \in M' \vee (a_{i}, t_{j}) \in M' \}$ and output $M$.
   \label{step:pk6-m}
  \end{algorithmic}
 \end{algorithm}
 
 \begin{algorithm}[htb]
  \caption{Translating instances}\label{alg:convi}
  \begin{algorithmic}[1]
   \REQUIRE An instance $I$ for MAX SMTI-1TM.
   \ENSURE An instance $I'$ for SMI.
   \STATE Let $X$ and $Y$ be the sets of men and women of $I$, respectively.
   \STATE Let $X':= \{a_{i} \mid m_{i} \in X\} \cup \{b_{j} \mid w_{j} \in Y\}$ be the set of men of $I'$.\label{step:convi6}
   \STATE Let $Y':= \{s_{j} \mid w_{j}\in Y\} \cup \{t_{j} \mid w_{j} \in Y\}$ be the set of women of $I'$.
   \STATE Each $a_{i}$'s list is constructed as follows: Consider a tie $(w_{j_1} \ w_{j_2} \ \cdots \ w_{j_k})$ in $m_{i}$'s list in $I$.  We assume without loss of generality that $j_1 < j_2 < \cdots < j_k$.  (If not, just arrange the order, which does not change the instance.)  Replace each tie $(w_{j_1} \ w_{j_2} \ \cdots \ w_{j_k})$ by a strict order of $2k$ women $t_{j_1} \ t_{j_2} \ \cdots \ t_{j_k} \ s_{j_1} \ s_{j_2} \ \cdots \ s_{j_k}$.  
A woman who is not included in a tie is considered as a tie of length one.  
   \STATE Each $b_{j}$'s list is defined as ``$b_{j}: s_{j} \ \ t_{j}$''. 
   \STATE For each $j$, let $P(w_{j})$ be the list of $w_{j}$ in $I$, and $Q(w_{j})$ be the list obtained from $P(w_{j})$ by replacing each man $m_{i}$ by $a_{i}$.  Then $s_{j}$ and $t_{j}$'s lists are defined as follows:
\[
\renewcommand\arraystretch{1.2}
   \begin{array}{cccccccccccccccc}
s_{j}: & Q(w_{j}) & b_{j} \\
t_{j}: & b_{j} & Q(w_{j})
\end{array}
   \]
  \end{algorithmic}
 \end{algorithm}

Now we start formal proofs for the correctness.
 
 \begin{lemma}\label{lm:stability}
  Algorithm~\ref{alg:pk6} always outputs a stable matching.
 \end{lemma}

 \begin{proof}
  Let $M$ be the output of Algorithm~\ref{alg:pk6} and $M'$ be the matching obtained at Step \ref{step:pk6gs} of Algorithm~\ref{alg:pk6}.  We first show that $M$ is a matching.
 Since $M'$ is a matching, $a_{i}$ appears at most once in $M'$, so $m_{i}$ appears at most once in $M$.
Observe that $b_{j}$ is matched in $M'$, as otherwise $(b_{j}, t_{j})$ blocks $M'$, contradicting the stability of $M'$ in $I'$.  Hence at most one of $s_{j}$ and $t_{j}$ can be matched with $a_{i}$ for some $i$, which implies that $w_{j}$ appears at most once in $M$.  Thus $M$ is a matching.

  We then show the stability of $M$.  
  Since $M'$ is the output of MGS, it is stable in $I'$.
  Now suppose that $M$ is unstable in $I$ and there is a blocking pair $(m_{i},w_{j})$ for $M$. There are four cases:
  \begin{description}
   \item[\boldmath Case (i): both $m_{i}$ and $w_{j}$ are single.]
	      Since $m_{i}$ is single in $M$, Step~\ref{step:pk6-m} of Algorithm~\ref{alg:pk6} implies that $a_{i}$ is single in $M'$.
%
Since $w_{j}$ is single in $M$, $s_{j}$ is not matched in $M'$ with anyone in $Q(w_{j})$, i.e., $s_{j}$ is single or matched with $b_{j}$.  Note that $(a_{i}, s_{j})$ is a mutually acceptable pair because $(m_{i},w_{j})$ is a blocking pair, and $a_{i} \succ_{s_{j}} b_{j}$ in $I'$ by construction.  Thus $(a_{i}, s_{j})$ blocks $M'$, a contradiction.

   \item[\boldmath Case (ii): $w_{j} \succ_{m_{i}} M(m_{i})$ and $w_{j}$ is single.]
	      Let $M(m_{i}) = w_{k}$.  Then, by construction of $M$, $M'(a_{i})$ is either $s_{k}$ or $t_{k}$.
By construction of $I'$, $w_{j} \succ_{m_{i}} w_{k}$ implies both $s_{j} \succ_{a_{i}} s_{k}$ and $s_{j} \succ_{a_{i}} t_{k}$, and in either case we have that $s_{j} \succ_{a_{i}} M'(a_{i})$ in $I'$.
	      Since $w_{j}$ is single in $M$, by the same argument as Case (i), $s_{j}$ is either single or matched with $b_{j}$ in $M'$.  Hence $(a_{i}, s_{j})$ blocks $M'$.

   \item[\boldmath Case (iii): $m_{i}$ is single and $m_{i} \succ_{w_{j}} M(w_{j})$.]
	      Since $m_{i}$ is single in $M$, $a_{i}$ is single in $M'$ by the same argument as Case (i).
	      Let $M(w_{j})=m_{k}$.  Then, by construction of $M$, either $s_{j}$ or $t_{j}$ is matched with $a_{k}$, and the other is matched with $b_{j}$ since $b_{j}$ can never be single as we have seen in an earlier stage of this proof.  That is, $M'(s_{j})$ is either $a_{k}$ or $b_{j}$.
Note that $m_{i} \succ_{w_{j}} m_{k}$ in $P(w_{j})$ implies $a_{i} \succ_{s_{j}} a_{k}$ in $Q(w_{j})$, so in either case $a_{i} \succ_{s_{j}} M'(s_{j})$ in $I'$.
	      Therefore $(a_{i}, s_{j})$ blocks $M'$.
   \item[\boldmath Case (iv): $w_{j} \succ_{m_{i}} M(m_{i})$ and $m_{i} \succ_{w_{j}} M(w_{j})$.]
By the same argument as Case (ii), we have that $s_{j} \succ_{a_{i}} M'(a_{i})$ in $I'$.  By the same argument as Case (iii), we have that $a_{i} \succ_{s_{j}} M'(s_{j})$ in $I'$.  Hence $(a_{i}, s_{j})$ blocks $M'$.
  \end{description}
 \end{proof}

 \begin{lemma}\label{lm:1.5}
  Algorithm~\ref{alg:pk6} always outputs a 1.5-approximate solution.
 \end{lemma}

 \begin{proof}
  Let $I$ be an input, $M_{opt}$ be a maximum stable matching for $I$, and $M$ be the output of Algorithm~\ref{alg:pk6}.
  We show that $\frac{|M_{opt}|}{|M|} \leq 1.5$.
  Let $G=(X \cup Y,E)$ be a bipartite (multi-)graph with vertex bipartition $X$ and $Y$, where $X$ corresponds to men and $Y$ corresponds to women of $I$.
  The edge set $E$ is a union of $M$ and $M_{opt}$, that is, $(m_{i},w_{j}) \in E$ if and only if $(m_{i},w_{j})$ is a pair in $M$ or $M_{opt}$.
  If $(m_{i},w_{j})$ is a pair in both $M$ and $M_{opt}$, then $E$ contains two edges $(m_{i},w_{j})$, which constitute a ``cycle'' of length two.  
  An edge in $E$ corresponding to $M$ (resp. $M_{opt}$) is called an $M$-edge (resp. $M_{opt}$-edge).  
  Since the degree of each vertex of $G$ is at most 2, each connected component of $G$ is an isolated vertex, a cycle, or a path.

  It is easy to see that $G$ does not contain a single $M_{opt}$-edge as a connected component, since if such an edge $(m_{i}, w_{j})$ exists, then $(m_{i}, w_{j})$ is a blocking pair for $M$, contradicting the stability of $M$. 
  In the following, we show that $G$ does not contain, as a connected component, a path of length three $m_{i}-w_{j}-m_{k}-w_{\ell}$ such that $(m_{i},w_{j})$ and $(m_{k},w_{\ell})$ are $M_{opt}$-edges and $(m_{k},w_{j})$ is an $M$-edge.  If this is true, then for any connected component $C$ of $G$, the number of $M$-edges in $C$ is at least two-thirds of the number of $M_{opt}$-edges in $C$, implying $\frac{|M_{opt}|}{|M|} \leq 1.5$.
 
  Suppose that such a path exists.
  Note that $m_{i}$ and $w_{\ell}$ are single in $M$.
  If $m_{i} \succ_{w_{j}} m_{k}$, then $(m_{i}, w_{j})$ blocks $M$.  Since women's preference lists do not contain ties, we have that $m_{k} \succ_{w_{j}} m_{i}$.  If $w_{\ell} \succ_{m_{k}} w_{j}$, then $(m_{k}, w_{\ell})$ blocks $M$.  If $w_{j} \succ_{m_{k}} w_{\ell}$, then $(m_{k}, w_{j})$ blocks $M_{opt}$.  Hence $w_{j}$ and $w_{\ell}$ are tied in $m_{k}$'s list.
  Then by construction of $I'$, (i) $t_{\ell} \succ_{a_{k}} s_{j}$.  (Hereafter, referring to Fig.~\ref{fig:part} would be helpful. Here, the order of $t_{j}$ and $t_{\ell}$ in $a_{k}$'s list is uncertain, i.e., it may be the opposite, but this order is not important in the rest of the proof.)
  Since $w_{\ell}$ is single in $M$, either $s_{\ell}$ or $t_{\ell}$ is single in $M'$.
  If $s_{\ell}$ is single in $M'$, then $(b_{\ell}, s_{\ell})$ blocks $M'$, a contradiction.  Hence (ii)  $t_{\ell}$ is single in $M'$.
  Since $M(m_{k}) = w_{j}$, either $M'(a_{k}) = s_{j}$ or $M'(a_{k}) = t_{j}$ holds.
  In the former case, (i) and (ii) above imply that $(a_{k}, t_{\ell})$ blocks $M'$,
  so assume the latter, i.e., $M'(a_{k}) = t_{j}$.
  Recall that either $s_{j}$ or $t_{j}$ is matched with $b_{j}$ in $M'$, so $M'(s_{j}) = b_{j}$.
  Since $(m_{i},w_{j})$ is an acceptable pair in $I$, we have that $a_{i} \succ_{s_{j}} b_{j}$.
  Since $m_{i}$ is single in $M$, $a_{i}$ is single in $M'$.  Hence 
  $(a_{i}, s_{j})$ blocks $M'$, a contradiction.

\begin{figure}[htb]
\begin{center}
\begin{tabular}{llllllllllllll}
$a_{i}$: & $\cdots$ \ $s_{j}$ \ $\cdots$ & \hspace{10mm} & $s_{j}$: & $\cdots$ \ $a_{i}$ \ $\cdots$ \ $b_{j}$ \\
$b_{i}$: & $s_{i}$ \ \ $t_{i}$ & \hspace{1mm} & $t_{j}$: & $b_{j}$ \ $\cdots$ \ $a_{k}$ \ $\cdots$  \\
&&&&\\
$a_{k}$: & $\cdots$ \ $t_{j}$ \ $\cdots$ \ $t_{\ell}$ \ $\cdots$ \ $s_{j}$ \ $\cdots$ & \hspace{1mm} & $s_{\ell}$: & $\cdots$ \ $b_{\ell}$ \\
$b_{k}$: & $s_{k}$ \ \ $t_{k}$ & \hspace{1mm} & $t_{\ell}$: & $b_{\ell}$ \ $\cdots$ \\
&&&&\\
$a_{\ell}$: & $\cdots$ \\
$b_{\ell}$: & $s_{\ell}$ \ \ $t_{\ell}$
\end{tabular}
\caption{A part of the preference lists of $I'$}\label{fig:part}
\end{center}
\end{figure}
 \end{proof}
 
 \begin{lemma}\label{lm:strategy-proof}
  Algorithm~\ref{alg:pk6} is a man-strategy-proof mechanism.
 \end{lemma}

 \begin{proof}
  The proof is similar to that of Lemma \ref{lm:positiveSMTI}.  Suppose that Algorithm~\ref{alg:pk6} is not a man-strategy-proof mechanism.
  Then there are MAX SMTI-1TM instances $I$ and $J$ and a
  man $m_{i}$ having the following properties: $I$ and
  $J$ differ in only $m_{i}$'s preference list, and $m_{i}$ prefers $M_{J}$ to $M_{I}$, 
  where $M_{I}$ and $M_{J}$ are the outputs of Algorithm~\ref{alg:pk6} for $I$ and $J$, respectively.
  Then either (i) $M_{J}(m_{i}) \succ_{m_{i}} M_{I}(m_{i})$ in $I$, or (ii) $m_{i}$ is single in $M_{I}$ and $M_{J}(m_{i})$ is acceptable to $m_{i}$ in $I$.

 Let $I'$ and $J'$ be the SMI-instances constructed by Algorithm~\ref{alg:convi}.
Since $I$ and $J$ differ in only $m_{i}$'s preference list, $I'$ and $J'$ differ in only $a_{i}$'s preference list.
 Let $M_{I'}$ and $M_{J'}$, respectively, be the outputs of MGS applied to $I'$ and $J'$.  In case of (i), we have that $M_{J'}(a_{i}) \succ_{a_{i}} M_{I'}(a_{i})$ in $I'$, due to Algorithm~\ref{alg:convi} and Step~\ref{step:pk6-m} of Algorithm~\ref{alg:pk6}.  In case of (ii), $a_{i}$ is single in $M_{I'}$ because $m_{i}$ is single in $M_{I}$, and $M_{J'}(a_{i})$ is acceptable to $a_{i}$ in $I'$ because $M_{J}(m_{i})$ is acceptable to $m_{i}$ in $I$.  This implies that $a_{i}$ has a successful strategy in $I'$, contradicting to the man-strategy-proofness of MGS for SMI \cite{gi89}.
 \end{proof}

By Lemmas \ref{lm:stability}, \ref{lm:1.5}, and \ref{lm:strategy-proof}, we can conclude that Algorithm~\ref{alg:pk6} is a man-strategy-proof 1.5-approximate-stable mechanism for MAX SMTI-1TM.
\end{proof}

\newpage

\appendix

\section{The Man-Oriented Gale-Shapley Algorithm}\label{sec:MGS}

During the course of the algorithm, each person takes one of two states ``free'' and  ``engaged''.  At the beginning, everyone is free and the matching $M$ is initialized to the empty set.  At one step of the algorithm, an arbitrary free man $m$ proposes to the top woman $w$ in his current list.  If $w$ is free, then $m$ and $w$ are provisionally matched and $(m,w)$ is added to $M$.  If $w$ is engaged and matched with $m'$, then $w$ compares $m$ and $m'$, takes the preferred one, and rejects the other.  The rejected man deletes $w$ from the list and becomes (or keeps to be) free.  When there is no free man, the matching $M$ is output.  The pseudo-code is given in Algorithm \ref{fig:MGS}.

\begin{algorithm}[htb]
\caption{The man-oriented Gale-Shapley algorithm}\label{fig:MGS}
 \begin{algorithmic}[1]
  \STATE Let $M:=\emptyset$ and all people be free.
  \WHILE{there is a free man whose preference list is non-empty}
  \STATE Let $m$ be any free man.
  \STATE Let $w$ be the woman at the top of $m$'s (current) list.
  \IF{$w$ is free}
  \STATE Let $M:= M \cup \{(m,w) \}$, and $m$ and $w$ be engaged.
  \ENDIF
  \IF{$w$ is engaged}
  \STATE Let $m'$ be $w$'s partner.
  \IF{$w$ prefers $m'$ to $m$}
  \STATE Delet $w$ from $m$'s list.
  \ELSE
  \STATE Let $M:= M \cup \{(m,w) \} \setminus \{(m',w) \}$.
  \STATE Let $m'$ be free and $m$ be engaged.
  \STATE Delete $w$ from $m'$'s list.
  \ENDIF
  \ENDIF
  \ENDWHILE
  \STATE Output $M$.
 \end{algorithmic}
\end{algorithm}

\section{Non-Strategy-Proofness of Existing 1.5-approximation Algorithms for MAX SMTI-1TM}\label{sec:counter-ex}

Kir\'aly \cite{k13} presented a 1.5-approximation algorithm for general MAX SMTI (i.e., ties can appear on both sides).  We modify it in the following two respects.

\begin{enumerate}
\item Men's proposals do not get into the 2nd round.

\item When there is arbitrarity, the person with the smallest index is prioritized. 
\end{enumerate}

\noindent
Ideas behind these modifications are as follows: For item 1, since there is no ties in women's preference lists, executing the 2nd round does not change the result.  The role of item 2 is to make the algorithm deterministic, so that the output is a function of an input (as we did in the proof of Lemma~\ref{lm:positiveSMTI}).  For completeness, we give a pseudo-code of the algorithm in Algorithm~\ref{fig:ALGO-NA}.  This algorithm is called ``New Algorithm'' in \cite{k13}, so we abbreviate it as NA here.

Each person takes one of three states, ``free'',  ``engaged'', and ``semi-engaged''.  Initially, all the persons are free.  At lines \ref{algo:propose1}, \ref{algo:propose2}, and \ref{algo:propose3}, man $m$ proposes to woman $w$.  Basically, the procedure is exactly the same as that of MGS.  If $w$ is free, then we let $M:= M \cup \{(m,w) \}$ and both $m$ and $w$ be engaged (we say $w$ {\em accepts} $m$).  If $w$ is engaged to $m'$ (i.e., $(m', w) \in M$) and if $m \succ_{w} m'$, then we let $M:= M \cup \{(m,w) \} \setminus \{(m',w) \}$, $m$ be engaged, and $m'$ be free.  We also delete $w$ from $m'$'s preference list (we say $w$ {\em accepts} $m$ and {\em rejects} $m'$).  If $w$ is engaged to $m'$ and $m' \succ_{w} m$, then we delete $w$ from $m$'s preference list (we say $w$ {\em rejects} $m$).

\begin{algorithm}[htb]
\caption{Kir\'aly's New Algorithm (NA) \cite{k13}}\label{fig:ALGO-NA}
 \begin{algorithmic}[1]
  \STATE Let $M:=\emptyset$ and all people be free.
  \WHILE{there is a free man whose preference list is non-empty}
  \STATE Among those men, let $m$ be the one with the smallest index.
  \IF{the top of $m$'s current preference list consists of only one woman $w$} 
  \STATE Let $m$ propose to $w$.\label{algo:propose1}
  \ENDIF
  \IF{the top of $m$'s current preference list is a tie} 
  \IF{all women in the tie are engaged}
  \STATE Among those women, let $w$ be the one with the smallest index.
  \STATE Let $m$ propose to $w$.\label{algo:propose2}
  \ENDIF
  \IF{there is a free woman in the tie}
  \STATE Among those free women, let $w$ be the one with the smallest index.
  \STATE Let $m$ propose to $w$.\label{algo:propose3}
  \ENDIF
  \ENDIF
  \ENDWHILE
  \STATE Output $M$.
 \end{algorithmic}
\end{algorithm}

There is an exception in the acceptance/rejection rule of a woman, when she receives the first and second proposals.  This is actually the key for guaranteeing 1.5-approximation, but this rule is not used in the subsequent counter-example so we omit it here.  Readers may consult to the original papers for the full description of the algorithm.

It is already proved that the (original) Kir\'aly's algorithm always outputs a stable matching and it is a 1.5-approximate solution, and it is not hard to see that the same results hold for the above NA for MAX SMTI-1TM.  However, as the example in Figures \ref{fig:true} and \ref{fig:false} shows, it is not a man-strategy-proof mechanism.

\begin{figure}[htb]
\begin{center}
\begin{tabular}{ccccccccccccc}
$m_{1}$: & $w_{2}$ & $w_{1}$ & & \hspace{1mm} & $w_{1}$: & $m_{2}$ & $m_{4}$ & $m_{1}$ \\
$m_{2}$: & ($w_{1}$ & $w_{3}$) & & \hspace{1mm} & $w_{2}$: & $m_{4}$ & $m_{1}$ & \\
$m_{3}$: & $w_{3}$ & & & \hspace{1mm} & $w_{3}$: & $m_{2}$ & $m_{3}$ & \\
$m_{4}$: & $w_{1}$ & $w_{2}$ & & \hspace{1mm} & $w_{4}$: & & &
\end{tabular}
\caption{A counter-example (true lists)}\label{fig:true}
\end{center}
\end{figure}

\begin{figure}[htb]
\begin{center}
\begin{tabular}{ccccccccccccc}
$m_{1}$: & $w_{1}$ & $w_{2}$ & & \hspace{1mm} & $w_{1}$: & $m_{2}$ & $m_{4}$ & $m_{1}$ \\
$m_{2}$: & ($w_{1}$ & $w_{3}$) & & \hspace{1mm} & $w_{2}$: & $m_{4}$ & $m_{1}$ & \\
$m_{3}$: & $w_{3}$ & & & \hspace{1mm} & $w_{3}$: & $m_{2}$ & $m_{3}$ & \\
$m_{4}$: & $w_{1}$ & $w_{2}$ & & \hspace{1mm} & $w_{4}$: & & &
\end{tabular}
\caption{A counter-example (manipulated by $m_{1}$)}\label{fig:false}
\end{center}
\end{figure}

If NA is applied to the true preference lists in Figure \ref{fig:true}, the obtained matching is $\{(m_{2}, w_{1}), (m_{3}, w_{3}), (m_{4}, w_{2}) \}$.  Suppose that $m_{1}$ flips the order of $w_{1}$ and $w_{2}$ (Figure \ref{fig:false}).  This time, NA outputs $\{(m_{1}, w_{2}), (m_{2}, w_{3}), (m_{4}, w_{1}) \}$ and $m_{1}$ successfully obtains a partner $w_{2}$.  By proposing to $w_{1}$ first, $m_{1}$ is able to let $m_{2}$ propose to $w_{3}$.  This allows $m_{4}$ to obtain $w_{1}$, which prevents $m_{4}$ from proposing to $w_{2}$.  This eventually makes it possible for $m_{1}$ to obtain $w_{2}$. 

We finally remark that the same example shows that the other two 1.5-approximation algorithms \cite{m09, p11} (with the tie-breaking rule 2 above) are not man-strategy-proof mechanisms either.

\end{document}